%% file: Paper_5.tex
\documentclass[twocolumn,showpacs,amsmath,amsfonts,amssymb,aps,showkeys]
              {revtex4}
\usepackage{latexsym,makeidx,ifthen,calc,longtable,graphics}
\usepackage{bm}

\begin{document}

\title{
Avoiding Haag's Theorem \\ 
with \\ 
Parameterized Quantum Field Theory
\footnote{The final publication is available at Springer 
via http://link.springer.com/article/10.1007/s10701-017-0065-8}
}

Version 3.4

\author{Ed Seidewitz}
\email{seidewitz@mailaps.org}
\affiliation{14000 Gulliver's Trail, Bowie MD 20720 USA}

\date{28 Febrary 2017}


\begin{abstract}

Under the normal assumptions of quantum field theory, Haag's theorem
states that any field unitarily equivalent to a free field must itself
be a free field.  Unfortunately, the derivation of the Dyson series
perturbation expansion relies on the use of the interaction picture,
in which the interacting field is unitarily equivalent to the free
field but must still account for interactions.  Thus, the traditional
perturbative derivation of the scattering matrix in quantum field
theory is mathematically ill defined.  Nevertheless, perturbative
quantum field theory is currently the only practical approach for
addressing scattering for realistic interactions, and it has been
spectacularly successful in making empirical predictions.  This paper
explains this success by showing that Haag's Theorem can be avoided
when quantum field theory is formulated using an invariant, fifth
\emph{path parameter} in addition to the usual four position
parameters, such that the Dyson perturbation expansion for the
scattering matrix can still be reproduced.  As a result, the
parameterized formalism provides a consistent foundation for the
interpretation of quantum field theory as used in practice and,
perhaps, for better dealing with other mathematical issues.

\end{abstract}

\maketitle


\input{Paper_5_Body}

\bibliography{../../RQMbib}

\end{document}

%% file: Paper_5_Body.tex
\noindent
\emph{Haag's theorem is very inconvenient; it means that the 
interaction picture exists only if there is no interaction.}
---Streater and Wightman \cite{streater64}

\section{Introduction} \label{sect:intro}

Haag's Theorem states that, under the normal assumptions of quantum
field theory (QFT), any field that is unitarily equivalent to a free
field must itself be a free field \cite{haag55,hall57,streater64}.
This is troublesome, because the usual Dyson perturbation expansion of
the scattering matrix is based on the interaction picture, in which
the interacting field is presumed to be related to the free field by a
unitary transformation.  And, as Streater and Wightman note, Haag's
theorem means that such a picture should not exist in the presence of
actual interaction.

There has therefore been some consternation in the literature over the
foundational implications of Haag's Theorem (in addition to
\cite{streater64}, see, for example, \cite{teller95,ticciati99,%
sklar00}).  Nevertheless, currently popular textbooks (such as
\cite{peskin95,weinberg95}) tend to simply ignore it.  Indeed, the
non-perturbative Lehmann-Symanzik-Zimmerman (LSZ) \cite{lsz55} and
Haag-Ruelle \cite{haag58,ruelle62} formalisms for scattering do not
run afoul of Haag's Theorem, so it is possible to formulate scattering
theory for QFT rigorously.

Unfortunately, the LSZ and Haag-Ruelle formalisms are not useful for
practical calculations of the $S$-matrix, for which perturbation
theory is always used.  So, the question remains as to why
perturbation theory works so well for this, despite Haag's Theorem
(see \cite{earman06} for a clear discussion of this point, and of
Haag's Theorem in general).

In this paper I will show that Haag's Theorem does not apply to a
\emph{parameterized} formulation of QFT. Not surprisingly, this
formulation starts from different assumptions than traditional QFT. A
parameterized formalism is one in which field operators have a fifth,
invariant parameter argument, in additional to the usual four position
arguments of Minkowski space.

There is a long history of approaches using a such a fifth parameter
for relativistic quantum mechanics, going back to proposals of Fock
\cite{fock37} and, particularly, Stueckelberg
\cite{stueckelberg41,stueckelberg42}, in the late thirties and early
forties.  The idea appeared subsequently in the work of a number of
well-known authors, including Nambu \cite{nambu50}, Feynman
\cite{feynman50,feynman51}, Schwinger \cite{schwinger51},
DeWitt-Morette \cite{morette51} and Cooke \cite{cooke68}.  However, it
was not until the seventies and eighties that the theory was more
fully developed, particularly by Horwitz and Piron
\cite{horwitz73,piron78} and Fanchi and Collins
\cite{collins78,fanchi78,fanchi83,fanchi93}, into what has come to be
called \emph{relativistic dynamics}.

A key feature of this approach is that time is treated comparably to
the three space coordinates, rather than as an evolution parameter.
The result is that relativistic quantum mechanics can be formulated in
a way that is much more parallel to non-relativistic quantum theory.
Further, the approach is particularly applicable to the study of
quantum gravity and cosmology, in which the fundamental equations
(such as the Wheeler-DeWitt equation) make no explicit distinction for
the time coordinate (see, e.g., \cite{teitelboim82,hartle83,hartle86,%
hartle95,halliwell01a,halliwell01b,halliwell02}).

Extension of this approach to a second-quantized QFT has been somewhat
more limited, focusing largely on application to quantum
electrodynamics \cite{fanchi79,saad89,land91,shnerb93,pavsic91,%
horwitz98,pavsic98}.  The formalism I will use here is derived from
related earlier work of mine on a foundational parameterized formalism
for QFT and scattering \cite{seidewitz06a,seidewitz06b,seidewitz09,%
seidewitz11}.  This formalism was developed previously in terms of
spacetime paths, but, in the present work, it is presented entirely in
field-theoretic mathematical language, without the use of spacetime
path integrals (though the concept of paths still remains helpful for
intuitive motivation).  The axiomatic approach used here is described 
in more detail in \refcite{seidewitz16}.

To show how Haag's Theorem formally breaks down, I will step through
an attempted proof of the equivalent of Haag's Theorem for
parameterized QFT, paralleling the proof of the theorem in traditional
axiomatic QFT. In order to provide a context for this,
\sect{sect:axiomatic:axioms} gives a brief overview of the axioms of
traditional QFT and \sect{sect:axiomatic:free} describes free fields
satisfying those axioms.  \sect{sect:axiomatic:haag} outlines the
formal proof of Haag's Theorem in that theory.

\Sect{sect:parameterized:axioms} then presents the parameterized QFT
formalism axiomatically, and \sect{sect:parameterized:free} describes
free fields satisfying those axioms.  (The axiomatic formalism for
parameterized QFT is only outlined here to the extent necessary for
the purposes of this paper---for a fuller presentation see
\refcite{seidewitz16}.)  \Sect{sect:parameterized:haag} next shows the
promised result, that Haag's Theorem does not apply in the axiomatic
parameterized formalism.  This allows a straightforward interaction
picture to be defined in \sect{sect:parameterized:interacting}, in
which interacting fields are related to free fields by a unitary
transformation.  \Sect{sect:conclusion} gives some concluding
thoughts.

In order to focus on the presentation of the key arguments in
\sect{sect:axiomatic:haag} and \sect{sect:parameterized:haag}, the proof of Haag's Theorem,
and its reconsideration, are just outlined semi-formally in these
sections.  However, the appendix includes formal statements of all
theorems, along with their proofs (or references to standard proofs 
in the literature).

Throughout, I will use a spacetime metric signature of $(-+++)$ and
take $\hbar = c = 1$.

\section{Axiomatic Quantum Field Theory} \label{sect:axiomatic}

The first formal proof of Haag's Theorem was based on axiomatic
quantum field theory \cite{hall57,streater64}.  This section
summarizes the axiomatic basis for QFT necessary to prove the theorem.
\Sect{sect:axiomatic:haag} then outlines the proof, which will be subsequently
reconsidered in the context of parameterized QFT.

\subsection{Axioms} \label{sect:axiomatic:axioms}

We begin with assumptions about the Hilbert space of states in
traditional QFT, and then continue with axioms on the fields that
operate on these states.

\setcounter{axiom}{-1}
\begin{axiom}[States] \label{axm:0}
    The states $\ket{\psi}$ are described by unit rays in a separable
    Hilbert space $\HilbH$, such that the following are true.
    \begin{enumerate}
	\item \emph{(Transformation Law)} Under a Poincar\'e
	transformation $\{\Delta x, \Lambda\}$ (where $\Delta x$ is a
	spacetime translation and $\Lambda$ is a Lorentz
	transformation), the states transformation according to a
	continuous, unitary representation $\opU(\Delta x, \Lambda)$
	of the Poincar\'e group.

	\item \emph{(Uniqueness of the Vacuum)} There is a unique,
	invariant vacuum state $\ket{0}$ such that $\opU(\Delta x,
	\Lambda)\ket{0} = \ket{0}$.

	\item \emph{(Spectral Condition)} Let 
	\begin{equation*}
	    \opU(\Delta x, 1) = \me^{\mi\opP^{\mu}\Delta x_{\mu}} \,.  
	\end{equation*}
	Then $\opP^{\mu}\opP_{\mu} = -m^{2}$, where $m$ is interpreted
	as a mass, and the eigenvalues of $\opP^{\mu}$ lie in or on
	the future light cone.
    \end{enumerate}
\end{axiom}

Note that the spectral condition given above captures the assumption 
that particles are \emph{on shell}, meeting the mass-shell 
constraint, $p^{2} = -m^{2}$, and that they advance into the future. 
In particular, time evolution is given by the (frame-dependent) 
Hamiltonian, $\opH = \opP^{0}$. 

QFT is then the theory of quantum fields that operate on these states.
Define a field operator $\oppsix$ that, for each spacetime position
$x$, acts on the states of the theory.  Actually, in order to avoid
improper states and delta functions, one should instead use the
so-called \emph{smeared} fields $\oppsi[f]$ such that
\begin{equation}
    \oppsi[f] = \intfour x\, f(x) \oppsix \,,
\end{equation}
for any $f$ drawn from an allowable set of \emph{test functions}.
However, for simplicity, the presentation here will freely use point
fields instead of the smeared fields, though the mathematics can still
be interpreted in terms of a more rigorous use of test functions and
distributions \cite{streater64,seidewitz16}.  Also, I only consider
scalar particles here, which is sufficient for addressing Haag's
Theorem.

The fields $\oppsix$ satisfy the following axioms, commonly known as
the \emph{Wightman axioms} \cite{streater64} (presented here for
fields that are not self-adjoint).

\renewcommand{\theaxiom}{\Roman{axiom}}
\begin{axiom}[Domain and Continuity of Fields] \label{axm:1}
    The field $\oppsix$ and its adjoint $\oppsit(x)$ are defined on a
    domain $D$ of states dense in $\HilbH$ containing the vacuum state
    $\ket{0}$.  The $\opU(\Delta x, \Lambda)$, $\oppsix$ and
    $\oppsit(x)$ all carry vectors in $D$ into vectors in $D$.
\end{axiom}

\begin{axiom}[Field Transformation Law] \label{axm:2}
    For any Poin\-car\'e transformation $\{\Delta x, \Lambda\}$, 
    \begin{equation} \label{eqn:A0a}\
	\opU(\Delta x, \Lambda) \oppsix
	\opU^{-1}(\Delta x, \Lambda)
	    = \oppsi(\Lambda x + \Delta x) \,.
    \end{equation}
\end{axiom}

\begin{axiom}[Local Commutivity]
    For any spacetime positions $x$ and $x'$,
    \begin{equation*}
	[\oppsix, \oppsi(x')] = [\oppsit(x), \oppsit(x')] = 0 \,.
    \end{equation*}
    Further, if $x$ and $x'$ are space-like separated,
    \begin{equation*}
	[\oppsix, \oppsit(x')] = 0 \,.
    \end{equation*}
\end{axiom}

\begin{axiom}[Cyclicity of the Vacuum] \label{axm:3}
    The vacuum state $\ket{0}$ is \emph{cyclic} for the fields
    $\oppsix$, that is, polynomials in the fields and their adjoints,
    when applied to the vacuum state, yield a set $D_{0}$ dense in
    $\HilbH$.
\end{axiom}

\subsection{Free Fields} \label{sect:axiomatic:free}

There is at least one class of field theories that satisfy the
Wightman axioms, those in which the fields are \emph{free}, that is,
they describe particles that do not interact.  The Hilbert space of a
free-field theory is thus a Fock space that is the direct sum of
subspaces of states with a fixed number of particles.  The fields
themselves are solutions of the Klein-Gordon equation
\begin{equation*}
    \left(-\frac{\partial^{2}}{\partial x^{2}} + m_{n}^{2} \right) 
	\oppsix = 0\,,
\end{equation*}
where $\partial^{2} / \partial x^{2}$ denotes the four-dimensional
relativistic d'Alembertian, with the two-point vacuum expectation
values
\begin{equation} \label{eqn:A0b}
    \bra{0}\oppsix\oppsit(\xz)\ket{0} = \propa \,,
\end{equation}
where
\begin{equation*}
    \propa \equiv (2\pi)^{-3}\intthree p\,
	\frac{\me^{\mi[-\Ep(x^{0}-\xz^{0}) + 
		       \threep\cdot(\threex-\threexz)]}}
	     {2\Ep}
\end{equation*}
and $\Ep \equiv \sqrt{\threep^{2}+m^{2}}$.

The usual approach for introducing interactions is to use the
\emph{interaction picture}, relating the interacting fields
$\oppsi_{I}(x)$ to the free fields by a transformation
\begin{equation*}
    \oppsi_{I}(t, \threex) = \opG(t)\oppsi(t, \threex)\opG(t)^{-1} \,,
\end{equation*}
where the time-dependence of $\opG$ is intended to reflect the
presence of interactions.  Such a $\oppsi_{I}(x)$ can be constructed
that satisfies the Wightman axioms (though its Lorentz covariance is
certainly not manifest, due to the time dependence of $\opG(t)$).
Unfortunately, Haag's Theorem shows that, under the Wightman axioms,
any field unitarily related to a free field will itself be a free
field, so an interacting theory cannot be consistently constructed in
this way.

\subsection{Haag's Theorem} \label{sect:axiomatic:haag}

This section outlines the proof of Haag's Theorem, which actually
follows from the results of two other theorems.  All three theorems
are summarized below.  \App{app:haag1} contains the formal statement
of these theorems, but the summary given here presents the key points
necessary for understanding Haag's theorem.

\begin{enumerate}
    \item Let $\oppsi_{1}$ and $\oppsi_{2}$ be two field operators,
    defined in respective Hilbert spaces $\HilbH_{1}$ and
    $\HilbH_{2}$, satisfying the Wightman axioms.  Suppose there exists
    a unitary operator $\opG$ such that, at a specific time $t$,
    $\oppsi_{2}(t, \threex) = \opG\oppsi_{1}(t, \threex)\opG^{-1}$.
    Then the equal-time vacuum expectation values of the fields at
    time $t$ are the same.

    \item Any free field satisfies \eqn{eqn:A0b}.  In addition,
    conversely, if the two-point expectation value for a field
    satisfies \eqn{eqn:A0b}, then the field is a free field.

    \item (Haag's Theorem) Let $\oppsi_{1}$ be a free field, so that
    it satisfies \eqn{eqn:A0b}, and let $\oppsi_{2}$ be a field
    unitarily related to $\oppsi_{1}$ at time $t$, as above.  Then
    $\oppsi_{2}$ will also satisfy \eqn{eqn:A0b} if $x$ and $\xz$ are
    both at time $t$.  However, the Lorentz-covariance of $\oppsi_{2}$
    allows the satisfaction of \eqn{eqn:A0b} to be extended to any two
    spacelike positions and then, by analytic continuation, to any two
    positions.  Therefore, $\oppsi_{2}$ is also a free field.
\end{enumerate}

Because of Haag's Theorem, a QFT of interacting particles that
satisfies the Wightman axioms cannot be unitarily equivalent to the
known theory of free particles.  So far, such an interacting theory
(for the case of four-dimensional spacetime) has yet to be found.
Nevertheless, as previously noted, empirically well-supported
predictions of the standard model are based on calculations using
perturbative QFT, which assumes that the interacting fields are, in
fact, unitarily related to the free fields, in contradiction to Haag's
Theorem.

This naturally leads one to consider whether there might be an
alternative set of axioms for QFT, for which Haag's Theorem would not
apply, but which could reproduce the results of traditional
perturbative QFT. In particular, note that the proof of \thm{thm:3}
essentially relies on a conflict between the presumption that the
fields are Lorentz-covariant and the special identification of time 
in the assumptions of \thm{thm:1}. This provides a motivation for 
considering parameterized QFT, which dispenses with any 
identification of time as an evolution parameter for fields.

With this motivation, I present, in the next section, a modified set
of axioms that are satisfied by a parameterized QFT, and I describe
the nature of parameterized free fields that satisfy those axioms.
Then I show that the proof of Haag's Theorem breaks down under the new
axioms, allowing for an interacting field that is unitarily equivalent
to a free field, and outline how such a field can be constructed using
an interaction picture in the parameterized theory.

\section{Parameterized Quantum Field Theory} \label{sect:parameterized}

This section presents the basic axioms of parameterized QFT, generally
paralleling the axioms for traditional QFT as presented in
\sect{sect:axiomatic}.  This is not intended to be a complete
presentation of the parameterized theory, but, instead, to just
outline a basic axiomatic foundation parallel to the overview given
for traditional QFT in \sect{sect:axiomatic}.  For a full presentation
of axiomatic, parameterized QFT, see \refcite{seidewitz16}.

\subsection{Axioms} \label{sect:parameterized:axioms}

We begin by stripping the assumptions on the Hilbert space of states 
down to the minimum necessary for a relativistic theory.

\renewcommand{\theaxiomp}{\arabic{axiomp}*}
\setcounter{axiomp}{-1}
\begin{axiomp}[States] \label{axm:0p}
    The states $\ket{\psi}$ are described by unit rays in a separable
    Hilbert space $\HilbH$.  Under a Poincar\'e transformation
    $\{\Delta x, \Lambda\}$ (where $\Delta x$ is a spacetime
    translation and $\Lambda$ is a Lorentz transformation), the states
    transformation according to a continuous, unitary representation
    $\opU(\Delta x, \Lambda)$ of the Poincar\'e group.
\end{axiomp}

Note, in particular, that there is no equivalent here of the spectral
condition.  That is, we begin, for the parameterized theory, with
states that are inherently \emph{off shell}.  Nevertheless, we can
still define the Hermitian operator $\opP$ such that
\begin{equation*}
    \opU(\Delta x,1) = \me^{\mi\opP \cdot \Delta x}
\end{equation*}
and interpret $\opP$ as the relativistic energy-momentum operator, as
usual.  However, we will \emph{not} require that the eigenvalues of
$\opP$ be on-shell.

Of course, $\opP^{\mu}\opP_{\mu}$ is still a Casimir operator of the
Poincar\'e group and, so, acts on any irreducible representation as an
invariant multiple of the identity.  Thus, in the parameterized
theory, we are \emph{not} requiring that the $\opU(\Delta x, \Lambda)$
be an irreducible representation for single-particle states.
Position-representation states may essentially be superpositions of
momentum states with squared-momentum $p^{2} = \mu$, for all $\mu$,
positive and negative.  This is sometimes considered explicitly as a
superposition of different mass states \cite{frastai95,land91,%
enatsu86} (see also the discussion at the end of Sec.  IIA of
\refcite{seidewitz06a}).

Moreover, we will no longer consider $\opP^{0}$, the generator of time
translation and the energy observable, to also be the Hamiltonian
operator.  Instead, we will allow any operator $\opH$ to be considered
a Hamiltonian operator if it has the following properties:
\begin{enumerate}
    \item $\opH$ is Hermitian.
    
    \item $\opH$ commutes with all spacetime transformations
    $\opU(\Delta x, \Lambda)$ (and, hence, with $\opP$).
    
    \item $\opH$ has an eigenstate $\ket{0} \in \HilbH$ such that
    $\opH\ket{0} = 0$, and this is the unique null (normalizable)
    eigenstate for $\opH$.
\end{enumerate}

Note that the concept of a \emph{vacuum state} $\ket{0}$ has
essentially been re-introduced here, but only relative to a choice of
Hamiltonian operator.  Each Hamiltonian operator must have a unique
vacuum state, but different Hamiltonian operators defined on the same
Hilbert space may have different vacuum states.  On the other hand,
the Hamiltonian operator is \emph{not} required to be positive
definite, since it is no longer considered to represent the energy
observable, so this vacuum state is not a ``ground state'' in the
traditional sense.  However, when $\HilbH$ is a Fock space with a
particle interpretation, the vacuum state for an identified
Hamiltonian will be the ``no particle'' state.

Given the choice of a Hamiltonian operator $\opH$, we can consider the
one-dimensional group generated by this operator.  The
frame-independent parameter $\lambda$ is the evolution parameter for
this group.  We can then define Schr\"odinger and Heisenberg pictures
for evolution in $\lambda$, analogously to the definitions for time
evolution,.

That is, consider a state $\ket{\psi} \in \HilbH$ and a Hamiltonian
operator $\opH$ defined on $\HilbH$.  The Schr\"odinger-picture
evolution of the state is then given by
\begin{equation*}
    \ket{\psi(\lambda)} = \me^{-\mi\opH\lambda} \ket{\psi} \,.
\end{equation*}
Similarly, if $\opA$ is an operator on $\HilbH$, then its 
Heisenberg-picture evolution is given by
\begin{equation*}
    \opA(\lambda) = \me^{\mi\opH\lambda}\opA\me^{-\mi\opH\lambda} \,.
\end{equation*}

As mentioned in \sect{sect:intro}, $\lambda$ can be considered to
parameterize the path of a particle in spacetime.  Physically, then, a
Hamiltonian operator $\opH$ is the generator of evolution of a
particle along its path.  But note that there is no constraint that
such a path have any fixed direction in time.  In particular, the path
is not constrained to be only a timelike trajectory, with $\lambda$
then being just the proper time for the particle.  Such a restriction
would be equivalent to re-introducing the mass-shell constraint,
which would essentially return the theory to the traditional
formulation.

Nevertheless, there is clearly an equivalence class of possible paths 
that a particle may take between fixed starting and ending points in 
spacetime, which need to be integrated over to obtain the complete 
amplitude for particle propagation from one point to another. And, 
after such integration, the path parameter should no longer appear in 
final, physical results. We will return to this point at the end of 
this section. 

For now, though, we turn back to the field theoretic formalism.
Fields in parameterized QFT are, in fact, defined in the same way as
fields in traditional QFT, as operators $\oppsix$, for each spacetime
position $x$, that act on the states of the theory.  The axioms are
also similar to those of traditional QFT. However, they are altered to
account for the lack of spectral constraints and the different
conception of the Hamiltonian.

\renewcommand{\theaxiomp}{\Roman{axiomp}*}
\begin{axiomp}[Domain and Continuity of Fields] \label{axm:1p}
    The field $\oppsix$ and its adjoint $\oppsit(x)$ are defined on a
    domain $D$ of states dense in $\HilbH$.  The $\opU(\Delta x,
    \Lambda)$, any Hamiltonian $\opH$, $\oppsix$ and $\oppsit(x)$ all
    carry vectors in $D$ into vectors in $D$.
\end{axiomp}

\begin{axiomp}[Field Transformation Law] \label{axm:2p}
    For any Poin\-car\'e transformation $\{\Delta x, \Lambda\}$, 
    \begin{equation*} 
	\opU(\Delta x, \Lambda) \oppsix 
	\opU^{-1}(\Delta x, \Lambda)
	    = \oppsi(\Lambda x + \Delta x) \,.
    \end{equation*}
\end{axiomp}

\begin{axiomp}[Commutation Relations] \label{axm:3p}
    For any spacetime positions $x$ and $x'$,
    \begin{equation*}
	[\oppsi(x'), \oppsix] = [\oppsit(x'), \oppsit(x)] = 0
    \end{equation*}
    and
    \begin{equation*}
	[\oppsi(x'), \oppsit(x)] = \delta^{4}(x'-x) \,.
    \end{equation*}
\end{axiomp}

\begin{axiomp}[Cyclicity of the Vacuum] \label{axm:4p}
    If $\opH$ is a Hamiltonian operator, then its vacuum state
    $\ket{0}$ is in the domain $D$ of the field operator, and
    polynomials in the field $\oppsix$ and its adjoint
    $\oppsit(x)$, when applied to $\ket{0}$, yield a set $D_{0}$ of
    states dense in $\HilbH$.
\end{axiomp}

These axioms are presented in terms of the essentially
Schr\"odinger-picture operators $\oppsix$.  However, for a Hamiltonian
$\opH$, we can now also define the (parameterized) Heisenberg-picture
form of a field operator:
\begin{equation*}
    \oppsixl 
	= \me^{\mi\opH\lambda}\oppsix\me^{-\mi\opH\lambda} \,.
\end{equation*}
Taking the limit of infinitesimal $\Delta \lambda$ then gives the
dynamic evolution of the field as
\begin{equation} \label{eqn:A1}
    \mi \pderiv{}{\lambda} \oppsixl = [\oppsixl, \opH] \,.
\end{equation}

\subsection{Free Fields} \label{sect:parameterized:free}

As for traditional QFT, we can define a parameterized field theory of
free particles that satisfies the above axioms.  And, again, the
Hilbert space of such a theory is a Fock space of states of a fixed
number of non-interacting particles.  Instead of the Klein-Gordon
equation, however, for a parameterized free field, we require that 
$\opH$ be chosen such that
\begin{equation} \label{eqn:A1a}
    [\oppsixl, \opH] 
	= \left(-\frac{\partial^{2}}{\partial x^{2}} + m^{2} 
	  \right) \oppsixl\,,
\end{equation}
where $m$ is the particle mass.  Together with \eqn{eqn:A1}, this
gives the \emph{Stueckelberg-Schr\"odinger field equation}
\cite{stueckelberg41,seidewitz06a}
\begin{equation} \label{eqn:A2}
    \mi \pderiv{}{\lambda}\oppsixl 
	= \left(-\frac{\partial^{2}}{\partial x^{2}} + m^{2} 
	  \right) \oppsixl\,.
\end{equation}

We can also define the momentum field $\oppsil{p}$ as the
four-dimensional Fourier transform of $\oppsixl$.  Like $\oppsixl$,
$\oppsil{p}$ evolves in $\lambda$ using $\opH$.  Taking the Fourier
transform of \eqn{eqn:A2} gives the field equation
\begin{equation} \label{eqn:A4}
    \mi \pderiv{}{\lambda}\oppsil{p} 
	= (p^{2} + m^{2}) \oppsil{p} \,.
\end{equation}
Recall that the four-momentum here is not necessarily on-shell.  Since
$p^{2}$ is relative to the Minkowski metric, $p^{2} + m^{2}$ can
take on any value, positive, negative or zero.

\Eqn{eqn:A4} may be integrated to get
\begin{equation} \label{eqn:A5}
    \oppsil{p} = \kernelpd \oppsilz{p} \,,
\end{equation}
where
\begin{equation*}
    \kernelpd \equiv 
	\me^{-\mi(p^{2} + m^{2})(\lambda - \lambdaz)} \,.
\end{equation*}
Taking the inverse Fourier transform of the momentum field in
\eqn{eqn:A5} then gives
\begin{equation} \label{eqn:A6}
    \oppsixl = \intfour \xz\, \kerneld \oppsilz{\xz} \,,
\end{equation}
where
\begin{equation} \label{eqn:A6a}
    \kerneld \equiv (2\pi)^{-4}
	\intfour p\, \me^{\mi p \cdot (x - \xz)} \kernelpd\,.
\end{equation}

The two-point vacuum expectation value of the parameterized free field
is given by
\begin{equation} \label{eqn:A6b}
    \bra{0} \oppsixl \oppsitlz{\xz} \ket{0} = \kerneld \,.
\end{equation}
This represents the propagation of a particle from four-position $\xz$
to four-position $x$ over any path $q$ with $q(\lambdaz) = \xz$ and
$q(\lambda) = x$. 

As mentioned earlier, paths in general are invariant under continuous
transformation of the path parameterization.  This is a well-known
gauge invariance of spacetime path formalisms.  After integrating out
the gauge volume of all possible functional re-parameterizations, the
only gauge variance that is left is the so-called \emph{intrinsic
length} of the path, $|\lambda - \lambdaz|$.  (For details see
\refcites{teitelboim82,seidewitz06a}.)

To represent propagation over all possible paths,it thus remains
necessary to integrate over all intrinsic lengths.  This may be
implemented by integrating \eqn{eqn:A6b} over all $\lambda >
\lambdaz$, which gives
\begin{equation} \label{eqn:A7}
    \int^{\infty}_{\lambdaz} \dl\, 
	\bra{0} \oppsixl \oppsitlz{\xz} \ket{0}
	= \prop \,,
\end{equation}
where
\begin{equation*}
    \begin{split}
	\prop &\equiv \int^{\infty}_{\lambdaz} \dl\, \kerneld \\
	      &= \int^{\infty}_{0} \dl\, 
			\kersym(x - \xz; \lambda) \\
	      &= -\mi(2\pi)^{-4}\intfour p\, 
			\frac{\me^{\mi p\cdot(x-\xz)}}
			     {p^{2}+m^{2}-\mi\varepsilon}
    \end{split}
\end{equation*}
is just the Feynman propagator \cite{seidewitz06a}.  Note that the
left-hand side of \eqn{eqn:A7} still involves the arbitrary parameter
value $\lambdaz$, even though the resulting propagator on the
right-hand side does not depend on it.  This simply reflects the
remaining global freedom to shift together the starting parameter
values of all paths.

As shown again by \eqn{eqn:A7}, particles are considered to be 
fundamentally off shell in the parameterized formalism. Nevertheless, 
it is always possible to divide the Feynman propagator into 
future-directed and past-directed parts (see, e.g., 
\refcite{peskin95} or \cite{weinberg95}):
\begin{equation*}
    \prop = \theta(x^{0}-\xz^{0})\propasym(x-\xz) 
	  + \theta(\xz^{0}-x^{0})\propasym(x-\xz)\conj \,,
\end{equation*}
where $\propa$ is the on-shell propagator from \eqn{eqn:A0b}.  Thus,
in the special case of a particle that is unambiguously
future-directed (i.e., $x^{0} > \xz^{0}$), $\prop = \propa$, and the
particle can be considered to be on shell.  This is essentially just
the reverse of the usual derivation in traditional QFT, in which
``physical'' particles are considered to be fundamentally on shell,
but in which future and past-directed on-shell propagators are combined
to form the Feynman propagator for off-shell ``virtual'' particles.

\subsection{Haag's Theorem Reconsidered} \label{sect:parameterized:haag}

This section reconsiders Haag's Theorem in the context of
parameterized QFT. The conclusion is that Haag's Theorem does not, in
fact, hold under the axioms given for parameterized QFT in
\sect{sect:parameterized:axioms}.  To see this, consider carrying out
a proof of Haag's Theorem for parameterized fields similar to the one
outlined in \sect{sect:axiomatic:haag} for traditional fields.  As in
\sect{sect:axiomatic:haag}, this argument is summarized below, with
the formalization given in \app{app:haag2}.

\begin{itemize}
    \item[1*.]  Let $\oppsi_{1}$ and $\oppsi_{2}$ be two off-shell
    field operators, defined in respective Hilbert spaces $\HilbH_{1}$
    and $\HilbH_{2}$, satisfying the axioms of parameterized QFT.
    Suppose there exists a unitary operator $\opG$ such that
    $\oppsi_{2}(x) = \opG \oppsi_{1}(x) \opG^{-1}$.  Then the
    equal-$\lambda$ vacuum expectation values of the fields are the
    same.

    \item[2*.]  Any free field satisfies \eqn{eqn:A6b}.  In addition,
    conversely, if the two-point expectation value for a field
    satisfies \eqn{eqn:A6b}, then the field is a free field.

    \item[3*.]  (Haag's Theorem?)  Let $\oppsi_{1}$ be a free field,
    so that it satisfies \eqn{eqn:A6b}, and let $\oppsi_{2}$ be a
    field unitarily related to $\oppsi_{1}$, as above.  Then
    $\oppsi_{2}$ will also satisfy \eqn{eqn:A6b} for $\lambda =
    \lambdaz$.  But now, if we try to generalize this to unequal
    parameter values $\lambda$ and $\lambdaz$, there is no equivalent
    of the Lorentz transformation to use in order to bring the
    parameter values back to equality.  All that is available is
    parameter translation, which would maintain the difference
    $\lambda - \lambdaz$.  Thus, \thm{thm:2p} does not apply and the
    proof of Haag's theorem does not go through.
\end{itemize}

\thm{thm:1p} here can be considered a kind of inverse to the
well-known Wightman reconstruction theorem (see \cite{streater64},
Theorem 3-7).  A consequence of the reconstruction theorem is that any
fields over spacetime that have the same vacuum expectation values
will be unitarily equivalent.  \thm{thm:1p} shows that fields that are
so unitarily related will (not surprisingly) have equal vacuum
expectation values.  (For parameterized fields, we do not require the
expectation values to meet a spectral condition or for the vacuum
state to be absolutely unique, but the proof of the reconstruction
theorem otherwise goes through, for a QFT that also does not meet such
conditions \cite{seidewitz16}.)

However, the key point here is that the assumptions of \thm{thm:1p}
and the resulting expectation values are all for parameterized fields
at \emph{equal values} of $\lambda = \lambdaz$.  If a similar
condition of equal expectation values also held for \emph{different}
values of $\lambda$ for the two fields, then the unitary operator
$\opG$ relating the fields would have to be independent of $\lambda$.
In this case, it is clear that, if one field satisfies the free field
equation \eqn{eqn:A2}, then the other will, too.

But the fact that \thm{thm:1p} only requires equality of expectation
values for equal values of $\lambda$ means that the fields can be
related by a \emph{different} value of $\opG(\lambda)$ for each value
of $\lambda$.  \thm{thm:3}, Haag's Theorem, shows that the requirement
of Lorentz-covariance of the fields prevents a similar construction in
traditional QFT in which $\opG$ depends on time (as required by the
traditional interaction picture).  In parameterized QFT, however, the
translation group of $\lambda$ does not mix with the group of Lorentz
transformations, so dependence of $\opG$ on $\lambda$ is not a
problem.

\subsection{Interacting Fields} \label{sect:parameterized:interacting}

In the absence of Haag's Theorem, we are free to use an interaction
picture representation for developing an interacting parameterized
QFT, which is what will be done in this section.  Begin by considering
a free field $\oppsix$ defined on a Hilbert space $\HilbH$ with a free
Hamiltonian operator $\opH$ and vacuum state $\ket{0}$.  Now consider
a unitary transformation $\opG$ of $\HilbH$ onto \emph{itself}.  This
then induces a transformed field $\oppsi'(x) = \opG\oppsix\opG^{-1}$
with a Hamiltonian $\opH' = \opG\opH\opG^{-1}$ and a vacuum state
$\ket{0'} = \opG\ket{0}$.

Now, in this case,
\begin{equation*}
    \begin{split}
	\paraml{\oppsi'}{x}
	    &= \me^{\mi\opH'\lambda}\oppsi'(x)
	       \me^{-\mi\opH'\lambda} \\
	    &= \me^{\mi\opH'\lambda}\opG\me^{-\mi\opH\lambda}
	       \oppsixl\me^{\mi\opH\lambda}\opG
	       \me^{-\mi\opH'\lambda} \\
	    &= \opG\oppsixl\opG^{-1} \,,
    \end{split}
\end{equation*}
since $\opH'\opG = \opG\opH$.  So, the Heisenberg-picture forms of the
fields are also unitarily related by $\opG$, which means that
$\paraml{\oppsi'}{x}$ is also a free field under the Hamiltonian
$\opH'$.

However, since $\oppsi'$ is defined on the same Hilbert space $\HilbH$
as $\oppsi$, it can be evolved equally well using $\opH$ rather than
$\opH'$.  Doing so results in the \emph{interaction-picture} form of
$\oppsi'$:
\begin{equation*}
    \paraml{\oppsi'_{I}}{x}
	= \me^{\mi\opH\lambda}\oppsi'(x)\me^{-\mi\opH\lambda}
	= \opG(\lambda)\oppsixl\opG^{-1}(\lambda) \,,
\end{equation*}
where
\begin{equation*}
    \opG(\lambda) = \me^{\mi\opH\lambda}G\me^{-\mi\opH\lambda} \,.
\end{equation*}
So the interaction-picture field $\paraml{\oppsi'_{I}}{x}$ is still 
unitarily related to the Heisenberg-picture free field $\oppsixl$, 
but now by a transformation $\opG(\lambda)$ that is \emph{not} 
constant in $\lambda$ (presuming that $\opG$ and $\opH$ do not 
commute).

The result is that the field equation for $\paraml{\oppsi'_{I}}{x}$ 
is no longer that of a free field. Consider that
\begin{equation} \label{eqn:B1}
    \begin{split}
	\pderiv{}{\lambda}\paraml{\oppsi'_{I}}{x}
	    &= \pderiv{\opG(\lambda)}{\lambda} \oppsixl  
	       \opG^{-1}(\lambda) \\
	    &\phantom{=}\ 
	       + \opG(\lambda) \pderiv{}{\lambda}\oppsixl
	         \opG^{-1}(\lambda) \\
	    &\phantom{=}\ 
	       - \opG(\lambda) \oppsixl \opG^{-1}(\lambda)
	         \pderiv{\opG(\lambda)}{\lambda}\opG^{-1}(\lambda) \\
	    &= \pderiv{\opG(\lambda)}{\lambda} \opG^{-1}(\lambda) 
	       \paraml{\oppsi'_{I}}{x} \\
	    &\phantom{=}\ 
	       + \mi\opG(\lambda)[\opH,\oppsixl]\opG^{-1}(\lambda) \\
	    &\phantom{=}\ 
	       - \paraml{\oppsi'_{I}}{x}
	         \pderiv{\opG(\lambda)}{\lambda}\opG^{-1}(\lambda) \\
	    &= \mi[\opH'(\lambda), \paraml{\oppsi'_{I}}{x}] \\
	    &\phantom{=}\ 
	       + [\pderiv{\opG(\lambda)}{\lambda}\opG^{-1}(\lambda),
	          \paraml{\oppsi'_{I}}{x}] \,,
    \end{split}
\end{equation}
where
\begin{equation*}
    \opH'(\lambda) 
	= \me^{\mi\opH\lambda}\opH'\me^{-\mi\opH\lambda}
	= \opG(\lambda)\opH\opG^{-1}(\lambda) \,.
\end{equation*}
But, since $\paraml{\oppsi'_{I}}{x}$ evolves according to $\opH$,
\begin{equation*}
    \pderiv{}{\lambda}\paraml{\oppsi'_{I}}{x}
	= \mi[\opH, \paraml{\oppsi'_{I}}{x}] \,,
\end{equation*}
so we can take
\begin{equation*}
    \opH'(\lambda) = \opH + \Delta\opH(\lambda) \,,
\end{equation*}
were
\begin{equation*}
    \Delta\opH(\lambda) 
	= \mi\pderiv{\opG(\lambda)}{\lambda}\opG^{-1}(\lambda) \,.
\end{equation*}

Now, as a free field, $\oppsixl$ satisfies the commutation condition 
with $\opH$ given in \eqn{eqn:A1a}. However, since $\opG(\lambda)$ 
does not depend on $x$, this also implies that
\begin{equation*}
    [\paraml{\oppsi'_{I}}{x}, \opH'(\lambda)]
	= \left(-\frac{\partial^{2}}{\partial x^{2}} + m^{2} 
	  \right) \paraml{\oppsi'_{I}}{x} \,.
\end{equation*}
Substituting this into \eqn{eqn:B1} then gives
\begin{multline*}
    \mi\pderiv{}{\lambda}\paraml{\oppsi'_{I}}{x}
	= \left(-\frac{\partial^{2}}{\partial x^{2}} + m^{2} 
	  \right) \paraml{\oppsi'_{I}}{x} \\
	  + [\Delta\opH(\lambda), \paraml{\oppsi'_{I}}{x}] \,.
\end{multline*}
Thus, the field equation for $\paraml{\oppsi'_{I}}{x}$ has the form 
of that for an \emph{interacting} field, with the interaction 
Hamiltonian $\Delta\opH(\lambda)$.

Note also that the $\lambda$-dependent transformation $\opG(\lambda)$ 
essentially induces a different effective vacuum state
\begin{equation*}
    \ketl{0'} = \opG(\lambda)\ket{0} = \me^{\mi\opH\lambda}\ket{0}
\end{equation*}
for each value of $\lambda$, such that $\opH'(\lambda)\ketl{0'} = 0$. 
However, consider the construction of multi-particle basis states 
from the interaction-picture field adjoints, all at a fixed value of 
$\lambda = \lambdaz$:
\begin{equation*}
    \begin{split}
	\ketlz{\argNx}_{I}
	    &= \paramlz{\oppsi'_{I}}{x_{1}}\adj \cdots
	       \paramlz{\oppsi'_{I}}{x_{N}}\adj \ketlz{0'}_{I} \\
	    &= \opG(\lambdaz) \oppsitlz{x_{1}} \cdots
	       \oppsitlz{x_{N}} \ket{0} \\
	    &= \opG(\lambdaz) \ketlz{\argNx} \,.
    \end{split}
\end{equation*}
Further, since $\ketlz{\argNx}_{I} \in \HilbH$, it can be expanded in 
the free-particle basis states for any $\lambda$, such that
\begin{multline*}
    \ketlz{\argNx}_{I} = \\
	\sum_{N'=0}^{\infty} 
	\left( \prod_{i=1}^{N}\intfour x'_{i} \right)
	\ket{\seqn{\xpli}{N'}} \\
	\bra{\seqn{\xpli}{N'}} \opG(\lambdaz) \ketlz{\argNx} \,,
\end{multline*}
where
\begin{equation*}
    \ket{\seqn{\xpli}{N'}}
	= \param{\oppsit}({x_{1}}{\lambda_{1}}) \cdots 
	  \param{\oppsit}({x_{N'}}{\lambda_{N'}}) \ket{0} \,.
\end{equation*}
The matrix elements 
\begin{multline*}
    G_{\lambdaz, \lambda_{1}, \ldots, \lambda_{N'}}
    (\argn{x'}{N'}; \argNx) = \\
	\bra{\seqn{\xpli}{N'}} \opG(\lambdaz) \ketlz{\argNx}
\end{multline*}
then represent the probability amplitude for $N$ particles at
spacetime positions $\argNx$ at the starting path parameter $\lambdaz$
to result, under interaction, in $N'$ particles at spacetime positions
$\argn{x'}{N'}$ at the ending path parameter values
$\argn{\lambda}{N'}$.

Now suppose that the interaction represented by $\opG(\lambda)$ is
restricted to a limited region of spacetime, and that the $x_{i}$ and
$x'_{i}$ are all outside this region.  Then $G_{\lambdaz, \lambda_{1},
\ldots, \lambda_{N'}} (\argn{x'}{N'}; \argNx)$ is the probability
amplitude for $N$ particles to propagate into the given region,
interact and result in $N'$ particles propagating out of the region.
Note that each of the outgoing particles is allowed a \emph{different}
ending parameter value.  This is important, because, as discussed at
the end of \sect{sect:parameterized:free}, it is necessary to
integrate over the ending parameter values to get the correct full
propagation factors for each particle:
\begin{multline} \label{eqn:B2}
    G_{\lambdaz}(\argn{x'}{N'}; \argNx) = \\
    \left( \prod_{i=1}^{N'} \int_{\lambdaz}^{\infty} \dl_{i} \right)
	G_{\lambdaz, \lambda_{1}, \ldots, \lambda_{N'}}
	(\argn{x'}{N'}; \argNx) \,.
\end{multline}

This, then, is just the probability amplitude for scattering from
$\argNx$ to $\argn{x'}{N'}$.  Now consider the limiting case in which
the $\argNx$ are in the infinite past and the $\argn{x'}{N'}$ in the
infinite future.  In this case, the initial propagation of a particle
from one of the $\argNx$ will always be future-directed, as will the
final propagation of a particle to one of the $\argn{x'}{N'}$.
Therefore, per the comment at the end of
\sect{sect:parameterized:free}, these propagations will be on shell,
as required for the incoming and outgoing external legs of the Feynman
diagram for an interaction.

Finally, take
\begin{equation*}
    \opG(\lambda) = \me^{-\mi\opV(\lambda)} \,,
\end{equation*}
where $\opV(\lambda)$ is an (appropriately integrated) product of
field operators and their adjoints, representing an individual
interaction vertex.  Expanding $\opG(\lambda)$ in a Taylor series then
gives a sum of Feynman diagrams, with vertices generated by
$\opV(\lambda)$ and Feynman propagators on internal edges.  Indeed, as
shown in detail in \refcite{seidewitz16}, the expression in
\eqn{eqn:B2} can be expanded to exactly duplicate, term for term, the
Dyson series for the scattering operator, as derived using
perturbation theory in traditional QFT.

But now the derivation from parameterized QFT does not suffer from the
mathematical inconsistency resulting from Haag's theorem.

\section{Conclusion} \label{sect:conclusion}

The original argument by Haag for what has become known as Haag's
Theorem was based on the requirement that the vacuum state $\ket{0}$
be the unique state invariant relative to Euclidean transformations
\cite{haag55}.  For the free theory, $\ket{0}$ is a null eigenstate of
the free Hamiltonian $\opH$.  Haag observed that the effective vacuum
state $\ket{0'}$ of the interacting field should also be invariant
under Euclidean transformations, which, under the assumption of the
uniqueness of $\ket{0}$ implies that $\ket{0'}$ equals $\ket{0}$ up to
a phase.  And, since $\ket{0'}$ is a null eigenstate of the
interacting Hamiltonian $\opH'$, it follows that $\ket{0}$ is, too.
However, interaction terms in $\opH'$ generally include the
interaction of the field with itself, such that $\opH'$ does
\emph{not} annihilate $\ket{0}$ (it ``polarizes the vacuum'').  This
is then a contradiction.

A similar argument could be made in the case of a parameterized
theory, if the assumption is made that the vacuum state $\ket{0}$ is
the unique state that is invariant under Poincar\'e transformations.
Then, since $\opG(\lambda)$ must transform as a Lorentz-invariant
scalar not dependent on position, the interacting vacuum $\ketl{0'} =
\opG(\lambda)\ket{0}$ would also be Poincar\'e invariant, and, thus,
equal to $\ket{0}$ up to a phase.  That means that $\ketl{0'}$ would
actually be independent of $\lambda$, which implies that
$\opG(\lambda)$ would commute with the free Hamiltonian $\opH$ and
$\oppsi'_{I}$ would be a free field.

In the parameterized formalism presented here, however, the free
vacuum $\ket{0}$ is \emph{not} required to be the unique
Poincar\'e-invariant state.  It only needs to be the unique
(normalizable) null eigenstate of the free Hamiltonian.  Unlike
$\ket{0}$, the parameterized interacting vacuum $\ketl{0'}$ depends
non-trivially on $\lambda$.  It is the null eigenstate of the
effective Hamiltonian $\opH'(\lambda)$ (which is itself dependent on
$\lambda$), not $\opH$, and can exist without compromising the
uniqueness of $\ket{0}$.

Put another way, the vacuum $\ket{0}$ is symmetric relative to
translations in $\lambda$.  This symmetry is preserved in the
interacting theory in the sense that the physics is not effected by a
translation in $\lambda$, and, in this sense, all $\ketl{0'}$ are
equivalent.  However, fixing on a specific $\lambda = \lambdaz$ breaks
the underlying symmetry for the interacting theory, choosing a
specific $\ketlz{0'}$, distinct from $\ket{0}$, as the effective
vacuum state for constructing interacting particle states.  In this
sense, the additional degree of freedom in the parameterized formalism
provides the possibility for the interacting vacuum to be different
from the non-interacting vacuum.

Thus, in parameterized QFT, it is possible for an interacting field to
be unitarily related to the corresponding free field.  And, as shown
in \refcite{seidewitz16}, it is possible to choose this transformation
so that the traditional Dyson perturbation expansion for scattering
amplitudes can be reproduced term by term in the new formalism.  This
explains why such an expansion works, despite Haag's Theorem---the
result in traditional QFT was essentially correct, only the derivation
was lacking.

Further, as argued in \cite{fraser06}, different formulations of QFT
may lead to different interpretations, even while being empirically
equivalent.  Clearly, one would like to base any interpretation on a
formulation that is rigorously defined mathematically.  But this is
problematic for traditional canonical quantum field theory, since
models of realistic interactions using the canonical formulation run
afoul of Haag's Theorem.

The parameterized formulation presented here resolves this problem.
Further, by allowing the Fock representation of a free field to be
extended to the corresponding interacting field, this approach allows
the intuitive particle interpretation of the free theory to be carried
over to the interacting theory.  Indeed, it can also provide for a
fuller interpretation in terms of spacetime paths, decoherence and
consistent histories over spacetime \cite{seidewitz06b,seidewitz11}.

Of course, this does not resolve all the mathematical issues with
traditional QFT, such as those involved in renormalization.  And the
approach still needs to be extended to cover gauge field theories and
non-Abelian interactions.  It is also worth noting that the
parameterized interacting vacuum states have superficial similarity to
the theta vacua that result from instantons in Yang-Mills theory (see,
for example, \refcites{ticciati99,bailin93}), which also have a $U(1)$
group symmetry.  Whether there is a deeper connection is a subject for
further investigation.

But, in any case, addressing the problem of Haag's Theorem is a step
toward building a firmer foundation, both mathematically and
interpretationally, for QFT in general.

\appendix
\section{Formalization}

This appendix presents the formal statements of Haag's Theorem and
related theorems, as considered in \sect{sect:axiomatic:haag} and
reconsidered in \sect{sect:parameterized:haag}.

\subsection{Haag's Theorem} \label{app:haag1}

As noted in \sect{sect:axiomatic:haag}, Haag's Theorem follows from two other
theorems, which I present here using notation consistent with that of
this paper, but without proof.  For details on the proofs, see
\cite{streater64}.

\begin{theorem} \label{thm:1}
    Let $\oppsi_{1}(x)$ and $\oppsi_{2}(x)$ be two field operators
    defined in respective Hilbert spaces $\HilbH_{1}$ and
    $\HilbH_{2}$.  Suppose there are continuous, unitary
    representations $\opU_{i}(\Delta\threex, R)$ of the inhomogeneous
    Euclidean group of translations $\Delta\threex$ and
    three-dimensional rotations $R$, defined on each $\HilbH_{i}$, for
    $i = 1,2$, such that, for a specific time $t$,
    \begin{equation*}
	\opU_{i}(\Delta\threex, R)\oppsi_{i}(t, \threex)
	    \opU_{i}^{-1}(\Delta\threex, R) =
	    \oppsi_{i}(t, R\threex + \Delta\threex) \,.
    \end{equation*}
    Suppose the representations possess unique invariant states 
    $\ket{0}_{i}$ such that
    \begin{equation*}
	\opU_{i}(\Delta\threex, R)\ket{0}_{i} = \ket{0}_{i} \,.
    \end{equation*}
    Suppose, finally, that there exists a unitary operator $\opG$ such 
    that, at time $t$,
    \begin{equation*}
	\oppsi_{2}(t, \threex) = 
	    \opG\oppsi_{1}(t, \threex)\opG^{-1} \,.
    \end{equation*}
    Then
    \begin{equation*}
	\opU_{2}(\Delta\threex, R) =
	    \opG\opU_{1}(\Delta\threex, R)\opG^{-1}
    \end{equation*}
    and
    \begin{equation*}
	c\ket{0}_{2} = \opG\ket{0}_{1} \,,
    \end{equation*}
    where $c$ is a complex number of modulus one.
\end{theorem}

This theorem immediately implies the following corollary.

\begin{corollary}
    In any two theories satisfying the hypotheses of \thm{thm:1}, the 
    equal-time vacuum expectation values are the same:
    \begin{multline*}
	{}_{1}\bra{0}\oppsi_{1}(t,\threex'_{1})\cdots
		\oppsi_{1}(t,\threex'_{N})
	    \oppsit_{1}(t,\threex_{1})\cdots
		\oppsit_{1}(t,\threex_{N})\ket{0}_{1} = \\
	{}_{2}\bra{0}\oppsi_{2}(t,\threex'_{1})\cdots
		\oppsi_{2}(t,\threex'_{N})
	    \oppsit_{2}(t,\threex_{1})\cdots
		\oppsit_{2}(t,\threex_{N})\ket{0}_{2} \,.
    \end{multline*}
\end{corollary}

The second theorem is a general result originally from \cite{jost61}.

\begin{theorem} \label{thm:2}
    If $\oppsi(x)$ is a field for which the vacuum is cyclic, and if
    \begin{equation*}
	\bra{0}\oppsi(x)\oppsit(\xz)\ket{0} = \propa \,,
    \end{equation*}
    then $\oppsi(x)$ is a free field.
\end{theorem}

Using these theorems, we can prove Haag's Theorem.

\begin{theorem}[Haag's Theorem] \label{thm:3}
    Suppose that $\oppsi_{1}(x)$ is a free field and $\oppsi_{2}(x)$
    is a local, Lorentz-covariant field.  Suppose further that the
    fields $\oppsi_{1}(x)$, $\dot{\oppsi}_{1}(x)$, $\oppsi_{2}(x)$ and
    $\dot{\oppsi}_{1}(x)$ satisfy the hypotheses of \thm{thm:1}.  Then
    $\oppsi_{2}(x)$ is also a free field.
\end{theorem}

\begin{proof}
    Since $\oppsi_{1}(x)$ is a free field, its vacuum expectation
    value is given by \eqn{eqn:A0b}.  The corollary to \thm{thm:1}
    then implies that, at a specific time $\tz$,
    \begin{equation} \label{eqn:A8}
	{}_{2}\bra{0}\oppsi_{2}(\tz,\threex)\oppsit_{2}
	    (\tz, \threexz)\ket{0}_{2} = \propasym(0, \threex - \threexz) \,.
    \end{equation}
    Any two position vectors $(t,x)$ and $(\tz,\xz)$ can be brought
    into the equal time plane $t = \tz$ by a Lorentz transformation,
    if their separation is spacelike.  Along with the given covariance
    of $\oppsi_{2}(x)$, this means that \eqn{eqn:A8} can be extended
    to any two spacelike positions and then, by analytic continuation,
    to any two positions:
    \begin{equation*}
	{}_{2}\bra{0}\oppsi_{2}(x)\oppsit_{2}(\xz)\ket{0}_{2} 
	    = \propa \,.
    \end{equation*}
    Haag's Theorem is then an immediate consequence of this and
    \thm{thm:2}.
\end{proof}

\subsection{Haag's Theorem Reconsidered} \label{app:haag2}

First, note that \thms{thm:1} and \ref{thm:2} may be directly adapted
for parameterized fields.  The axioms of the parameterized theory,
though, do not require a unique vacuum state, so we must simply assume
the relationship between vacuum states of the two fields in
\thm{thm:1p} below.  This simplifies the derivation of the relevant
conclusion of equality of expectation values.

\begin{theoremp} \label{thm:1p}
    Let $\oppsi_{1}(x)$ and $\oppsi_{2}(x)$ be two (off-shell,
    Schr\"odinger-picture) field operators, defined in respective
    Hilbert spaces, and $\opH_{1}$ and $\opH_{2}$ be Hamiltonian
    operators defined on those Hilbert spaces, with vacuum states
    $\ket{0}_{1}$ and $\ket{0}_{2}$.  Suppose, that there exists a
    unitary operator $\opG$, such that
    \begin{equation*}
	\oppsi_{2}(x) = \opG\oppsi_{1}(x)\opG^{-1}
    \end{equation*}
    and
    \begin{equation*}
	\ket{0}_{2} = \opG\ket{0}_{1} \,.
    \end{equation*}
    Then the equal-$\lambda$ vacuum expectation values of the fields
    are the same:
    \begin{equation*}
	\begin{split}
	    &{}_{1}\bra{0}\paraml{\oppsi_{1}}{x'_{1}}\cdots
	             \paraml{\oppsi_{1}}{x'_{N}}
	             \paraml{\oppsi_{1}\dadj}{x_{1}}\cdots
		     \paraml{\oppsi_{1}\dadj}{x_{N}}\ket{0}_{1} \\
	    &= {}_{2}\bra{0}\paraml{\oppsi_{2}}{x'_{1}}\cdots
	             \paraml{\oppsi_{2}}{x'_{N}}
	             \paraml{\oppsi_{2}\dadj}{x_{1}}\cdots
		     \paraml{\oppsi_{2}\dadj}{x_{N}}\ket{0}_{2} \,.
	\end{split}
    \end{equation*}
\end{theoremp}
\begin{proof}
    With the given assumptions, the expectation values are clearly
    equal for the Schr\"odinger-picture fields.  The parameterized
    Heisenberg-picture fields are related to the Schr\"odinger picture
    fields by the unitary transformations $\exp(-\mi\opH_{i}\lambda)$,
    and the vacuum states $\ket{0}_{i}$ are invariant under the
    corresponding such transformations.  Therefore, the vacuum
    expectation values of the parameterized fields are equal to the
    corresponding vacuum expectation values for the respective
    unparameterized fields and, hence, to each other, for all
    $\lambda$.
\end{proof}

\begin{theoremp} \label{thm:2p}
    If $\oppsixl$ is a field with a Hamiltonian $\opH$ having vacuum
    state $\ket{0}$, and if
    \begin{equation*}
	\bra{0}\oppsixl\oppsitlz{\xz}\ket{0} = \kerneld \,,
    \end{equation*}
    where $\kerneld$ is given by \eqn{eqn:A6}, then $\oppsixl$ is a
    free field.
\end{theoremp}
\begin{proof}
    Let $\braxl =\bra{0}\oppsil{x}$. Then,
    \begin{equation*}
	\begin{split}
	    \braxl &= \intfour \xz\, \innerxlxlz \braxlz \\
		   &= \intfour \xz\, \kerneld \braxlz \,.
        \end{split}
    \end{equation*}
    Therefore,
    \begin{equation*}
	\mi\pderiv{}{\lambda}\braxl
	    = \left( -\frac{\partial^{2}}{\partial x^{2}}
	             + m^{2} \right) \braxl \,.
     \end{equation*}
     However,
     \begin{equation*}
	 \begin{split}
	     \mi\pderiv{}{\lambda}\braxl
		 &= \mi\pderiv{}{\lambda}\bra{0}\oppsil{x}
		  = \bra{0} [\oppsil{x}, \opH] \\
		 &= \bra{0} \oppsil{x} \opH
		  = \braxl \opH \,,
	 \end{split}
     \end{equation*}
     since $\bra{0}\opH = 0$.  Thus, the Hamiltonian $\opH$ acts as
     \begin{equation*}
	 \opH =  -\frac{\partial^{2}}{\partial x^{2}} + m^{2} \,,
     \end{equation*}
     which is the position representation of the required free-field 
     Hamiltonian (see \sect{sect:parameterized:free}).
     
     Further, the equal-$\lambda$ expectation value for $\oppsi$ is
     \begin{equation*}
	 \bra{0}\oppsixl\oppsitl{\xz}\ket{0} 
	     = \kersym(x - \xz; 0)
	     = \delta^{4}(x - \xz) \,.
     \end{equation*}
     This, together with \axm{axm:3p}, implies that 
     \begin{equation*}
	 \bra{0}\oppsitl{\xz} \oppsixl \ket{0} = 0 \,, 
     \end{equation*}
     for all $x$ and $\xz$. In 
     particular, in the limit of $\xz \rightarrow x$, 
     \begin{equation*}
	 \bra{0} \oppsitl{x} \oppsixl \ket{0} 
	     = |\oppsixl \ket{0}|^{2} 
	     = 0 \,,
     \end{equation*}
     so $\oppsixl\ket{0} = 0$, for all $x$ and $\lambda$. This and the 
     commutation relations of \axm{axm:3p} indicate that $\oppsi(x)$ 
     and its adjoint can be used to build the states of a 
     free-particle Fock space.
\end{proof}

Now, even given \thms{thm:1p} and \ref{thm:2p}, the following
proposition corresponding to \thm{thm:3} turns out \emph{not} to be
true.

\setcounter{propositionp}{\value{theoremp}}
\begin{propositionp}[Haag's Theorem?]
    Suppose that $\paraml{\oppsi_{1}}{x}$ is a free field and
    $\paraml{\oppsi_{2}}{x}$ is a local, Lorentz-covariant field.
    Suppose further that the fields $\paraml{\oppsi_{1}}{x}$ and
    $\paraml{\oppsi_{2}}{x}$ satisfy the hypotheses of \thm{thm:1p}.
    Then $\paraml{\oppsi_{2}}{x}$ is also a free field.
\end{propositionp}

Given the assumptions of this proposition and \thm{thm:1p}, we can
easily deduce the equivalent of \eqn{eqn:A8}:
\begin{equation*}
    {}_{2}\bra{0}\paramlz{\oppsi_{2}}{x}
	\paramlz{\oppsi_{2}\adj}{\xz}\ket{0}_{2} 
	    = \kersym(x - \xz; 0) = \delta^{4}(x - \xz) \,.
\end{equation*}
However, just as \eqn{eqn:A8} was at the single time $\tz$, the
equivalent equation for the parameterized theory is at the single
parameter value $\lambdaz$.  But now, if we try to generalize this to
unequal parameter values $\lambda$ and $\lambdaz$, there is no
equivalent of the Lorentz transformation to use in order to bring the
parameter values back to equality.  All that is available is parameter
translation, which would maintain the difference $\lambda - \lambdaz$.

Therefore, it is possible for $\paraml{\oppsi_{1}}{x}$ and
$\paraml{\oppsi_{2}}{x}$ to have the same equal-$\lambda$ two-point
vacuum expectation value, but for their unequal-$\lambda$ expectation
values to differ.  Thus, \thm{thm:2p} does not apply and the proof of
Haag's theorem does not go through.